\documentclass[11pt]{article}
\pdfoutput=1
\usepackage{fullpage}
\usepackage{graphicx}					
\usepackage[utf8]{inputenc}
\usepackage{enumerate}
\usepackage{amssymb,amsmath,amsthm,amsfonts}
\usepackage{authblk}
\usepackage{color}
\usepackage{caption}
\usepackage{subcaption}
\newtheorem{claim}{Claim}

\newtheorem{theorem}{Theorem}[section]

\newtheorem{lemma}[theorem]{Lemma}

\newcommand{\tyler}[1]{\textcolor{red}{\textbf{Tyler: } #1}}

\newcommand{\toapp}[1]{{}}   % blanks out material that is sent to appendix
\newcommand{\old}[1]{{}}  % old material commented out

\newcommand{\referee}[1]{{}}  % old material commented out

\newcommand{\full}[1]{{#1}}  % to include in full version

\newcommand{\short}[1]{{}}  % to include in short version

\begin{document}

\title{TSP With Locational Uncertainty: The Adversarial Model}
%tyler: Try to get rid of the ``and'' %%%%%%%%%%%%%%%%%%%%
\author[2]{Gui Citovsky}
\author[1]{Tyler Mayer}
\author[1]{Joseph S. B. Mitchell}
\affil[1]{Dept. of Applied Mathematics and Statistics, Stony Brook University.\\
\texttt{\{tyler.mayer, joseph.mitchell\}@stonybrook.edu}}
\affil[2]{Google Manhattan.\\
\texttt{gcitovsky@gmail.com}}
\date{}

\maketitle

\begin{abstract}
In this paper we study a natural special case of the Traveling
Salesman Problem (TSP) with point-locational-uncertainty which we will
call the {\em adversarial TSP} problem (ATSP).  Given a metric space
$(X, d)$ and a set of subsets $R = \{R_1, R_2, ... , R_n\} : R_i
\subseteq X$, the goal is to devise an ordering of the regions,
$\sigma_R$, that the tour will visit such that when a single point is
chosen from each region, the induced tour over those points in the
ordering prescribed by $\sigma_R$ is as short as possible.  Unlike the
classical locational-uncertainty-TSP problem, which focuses on
minimizing the expected length of such a tour when the point within
each region is chosen according to some probability distribution,
here, we focus on the {\em adversarial model} in which once the choice
of $\sigma_R$ is announced, an adversary selects a point from each
region in order to make the resulting tour as long as possible.  In
other words, we consider an offline problem in which the goal is to
determine an ordering of the regions $R$ that is optimal with respect
to the ``worst'' point possible within each region being chosen by an
adversary, who knows the chosen ordering. We give a $3$-approximation
when $R$ is a set of arbitrary regions/sets of points in a metric
space.  We show how geometry leads to improved constant factor
approximations when regions are parallel line segments of the same
lengths, and a polynomial-time approximation scheme (PTAS) for the
important special case in which $R$ is a set of disjoint unit disks in
the plane.
\end{abstract}

\referee{Rvr 2: - I am not sure if the term "adversarial model" is so good, because it (at least for me) suggests that an online version of TSP is studied. Maybe either choose a different name or make sure that this is ruled out. Be aware that more than half of your audience are not native English speakers!}

\referee{Rvr 3: Throughout the paper, sometimes length and sometimes weight is used to refer to the weight/length of tours and cycles. Please be consistent.}
%% Joe: I have tried to make consistent. I used ``weight'' for the graph $\hat{G}$, but ``length'' when referring to the metric space

\old{\tyler{Abstract: $3$ apx only holds for even regions ?3 + 1/n? }}

\section{Introduction}

We consider the travelling salesperson problem (TSP) on uncertain sites.  We are given as input a set of $n$ uncertainty regions $R = \{R_1, R_2,\ldots, R_n\}$,  each of which is known to contain exactly one site that must be visited by the tour.  In the standard TSP, the regions $R_i$ are singleton points.  In the {\em TSP with neighborhoods} (TSPN), or {\em one-of-a-set TSP}, model, the goal is to compute an optimal tour that visits some point of each region $R_i$, and we are allowed to pick any point $p_i\in R_i$ to visit, making this choice in the most advantageous way possible, to minimize the length of the resulting tour that we compute.  In models of TSP with locational uncertainty, the regions $R_i$ model the support sets of probability distributions for the uncertain locations of the (random variable) sites $p_i$. The objective, then, may be to optimize some statistic of the tour length; e.g., we may wish to minimize the expected tour length, or minimize the probability that the tour length is greater than some threshold, etc.  In this paper, we study the version of the stochastic TSP model in which our goal is to optimize for the {\em worst case} choice of $p_i$ within each $R_i$.  We call this problem the {\em adversarial TSP}, or ATSP, as one can think of the choice of $p_i$ within each $R_i$ as being made by an adversary. Our goal is to compute a permutation $\sigma_R$ on the regions $R_i$ so that we minimize the length of the resulting tour on the points $p_i$, assuming that an adversary makes the choice of $p_i\in R_i$, given our announced permutation $\sigma_R$ on the regions.  While the TSPN seeks an optimal tour for the {\em best} choices of $p_i\in R_i$, the ATSP seeks an optimal tour for the {\em worst} choices of $p_i\in R_i$.

Another motivation for the ATSP solution is that one may seek a single permutation of the set of input sets $R_i$ so that the permutation is ``good'' (controls the worst-case choices of $p_i\in R_i$) for any of the numerous ($|R_1|\cdot|R_2|\cdot |R_3| \cdots |R_n|$) instances of TSP associated with the sets $R_i$, thereby avoiding repeated computations of TSP tours.  In certain vehicle routing applications, it may also be beneficial to establish a fixed ordering of visits to clients, even if the specific locations of these visits may vary in the sets $R_i$.
Further, in locationally uncertain TSP one may expect that probability distributions over the regions $R_i$ are imperfect and not known precisely, and that customer locations are known imprecisely (possibly for privacy concerns, with deliberate noise added for protecting the identity/privacy of users), making it important to optimize over all possible choices of site locations.

In Figure ~\ref{fig:tspc-order} we give a simple example showing that
the ordering given by a TSP on center points
($TSP_c$) can be suboptimal, by at least a factor of $\sqrt{2}$.  The input $R$ is a $\sqrt{n} \times
\sqrt{n}$ grid of vertical unit-length line segments with distance $1$
between midpoints of horizontally adjacent segments and with distance
$1+\epsilon$ between midpoints of vertically adjacent segments.  In Figure \ref{fig:tspn-order} we show that the ordering prescribed by a TSPN over the input regions can be at least a factor $2$ away from optimal.  The input is a set of $n$ segments, $n/2$ of which have length 1, and the remainder have length $\epsilon$; they are arranged in alternating order radially around a point or the boundary of a small circle.
In this case, the TSPN and the $TSP_c$ orderings give
  constant-factor approximations for the adversarial TSP; we will
  discuss this property later.

\begin{figure}
\begin{subfigure}{0.3\textwidth}
\includegraphics[scale = 0.3]{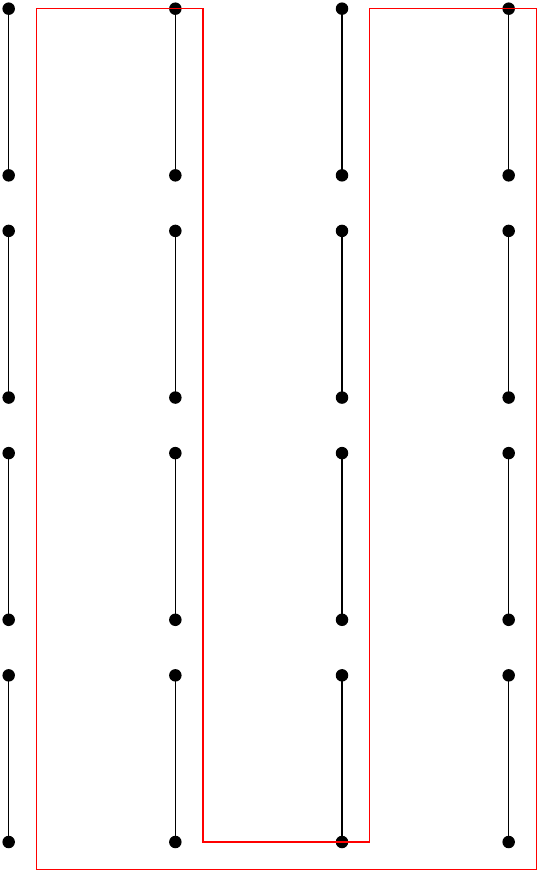}
\caption{Optimal solution $\approx n$}
\label{fig:opt_segs}
\end{subfigure}
\begin{subfigure}{0.3\textwidth}
\includegraphics[scale = 0.3]{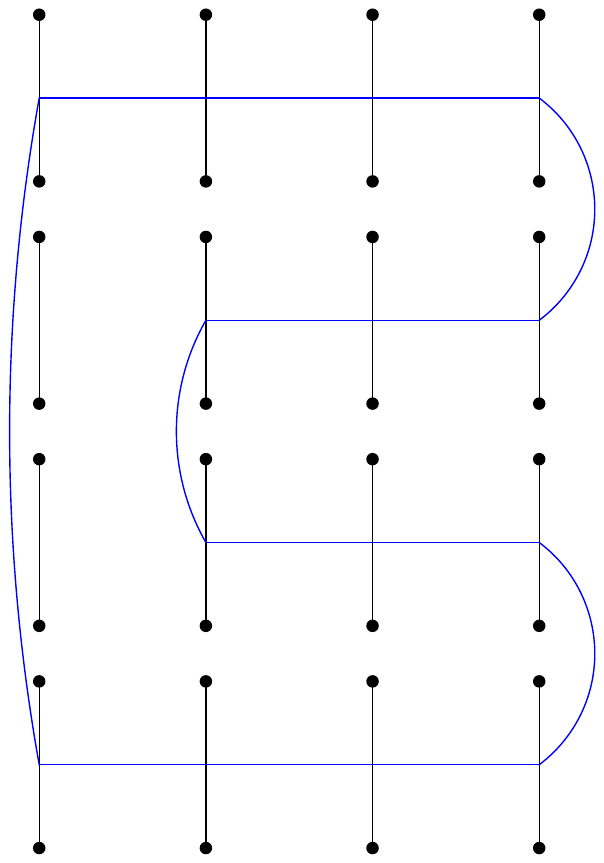}
\caption{$TSP_c$}
\label{tspc}
\end{subfigure}
\begin{subfigure}{0.3\textwidth}
\includegraphics[scale = 0.3]{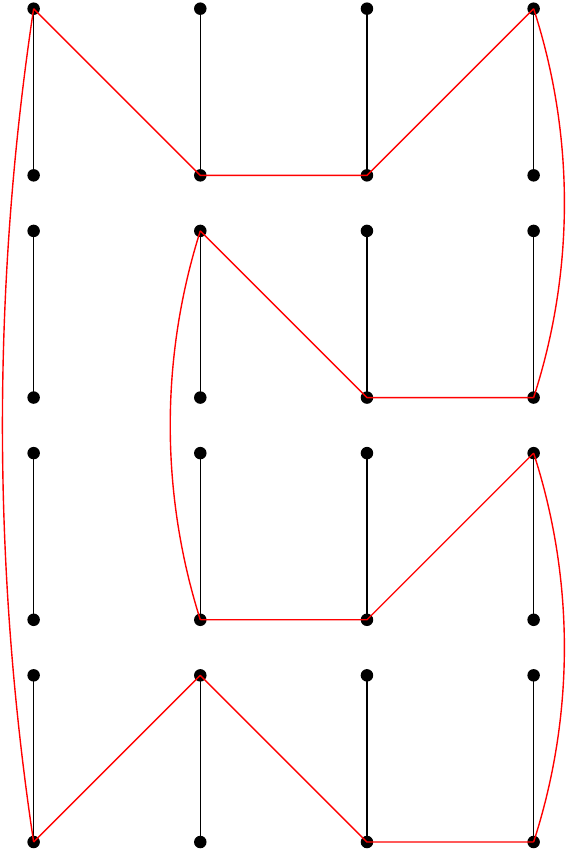}
\caption{Adversarial tour on $TSP_c$ ordering $\approx \sqrt{2}n$}
\label{fig:atspc_segs}
\end{subfigure}
\caption{TSP on center points ordering does not always provide an optimal solution to ATSP.}
\label{fig:tspc-order}
\end{figure}

\begin{figure}
\begin{subfigure}{0.3\textwidth}
\includegraphics[scale = 0.3]{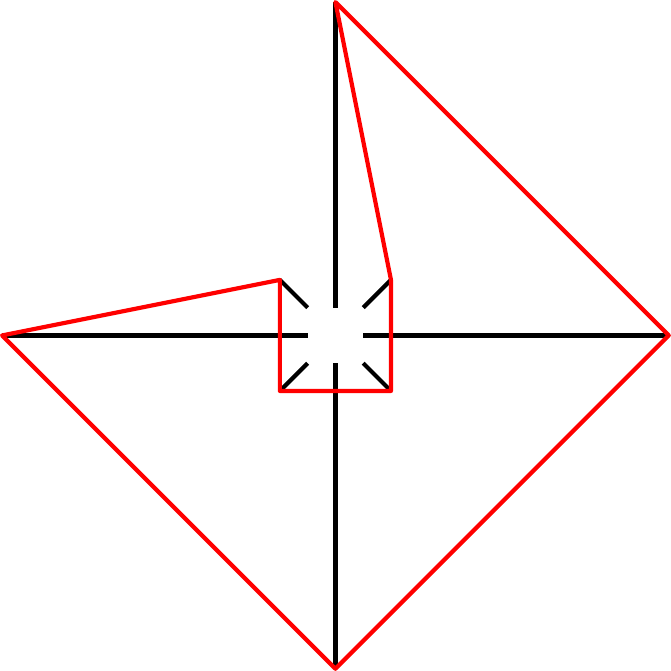}
\caption{Optimal solution $\approx n/2$}
\label{fig:tspn_opt_segs}
\end{subfigure}
\begin{subfigure}{0.3\textwidth}
\includegraphics[scale = 0.3]{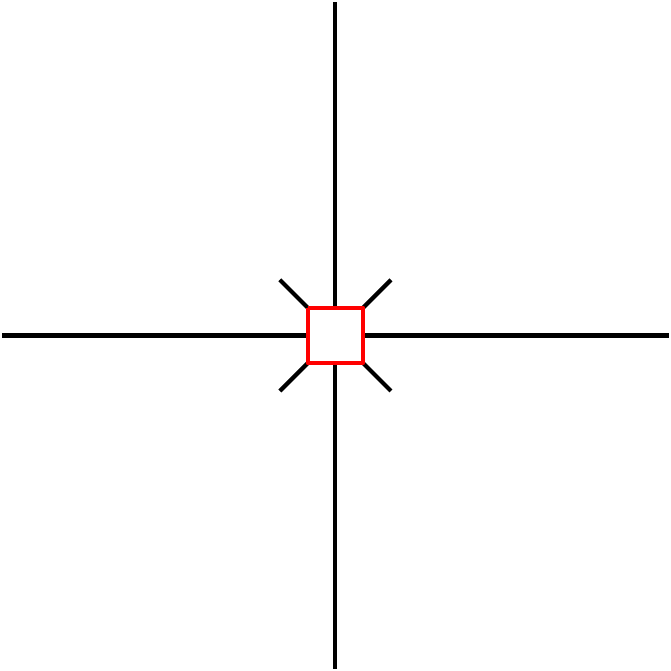}
\caption{TSPN}
\label{fig:tspn_segs}
\end{subfigure}
\begin{subfigure}{0.3\textwidth}
\includegraphics[scale = 0.3]{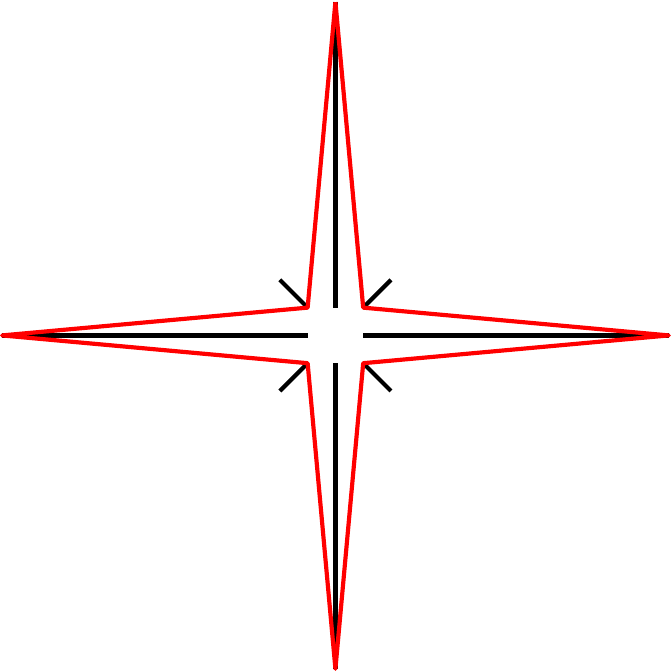}
\caption{Adversarial tour on TSPN ordering $\approx n$}
\label{fig:atspn_segs}
\end{subfigure}
\caption{TSPN ordering does not always provide an optimal solution to ATSP.}
\label{fig:tspn-order}
\end{figure}

In this paper, we initiate the study of the ATSP.  We give a
$3$-approximation when $R$ is a set of arbitrary regions/sets of
points in a metric space.  We exploit geometry to give an improved
approximation bound for the case of regions that are unit line
segments of the same orientation in the plane; we compute a
permutation with adversarial tour length at most $(7/3 +
\epsilon)|OPT| + 1$, where $|OPT|$ is the length of an optimal
solution. We further exploit geometry to give a polynomial time
approximation scheme (PTAS) for the important special case when $R$ is
a set of disjoint unit disks in the plane.

\subsubsection*{Related Work}
Geometric problems on imprecise points have been the subject of many recent investigations. L{\"o}ffler et al. \cite{loffler2010largest} study, given a set of $n$ uncertainty regions in the plane,  the problem of {\em selecting} a single point within each  region so that the area of the resulting convex hull is as large/small as possible.  They show a number of results, including an $O(n^3)$-time and an $O(n^7)$-time exact algorithm for maximizing the area of the convex hull of selected  points when the uncertainty regions are parallel line segments and disjoint axis aligned squares respectively.  They show that this problem is NP-Hard when the regions are line segments with arbitrary orientations.  
In the same paper, L{\"o}ffler et al. show that the problem of selecting a point within each region so that the resulting minimum spanning tree over those points is as small as possible is NP-Hard when the uncertainty regions are overlapping  disks and squares.  In his thesis \cite{fraser2012algorithms}, Fraser extends the prior minimum spanning tree result to show that the problem is still NP-Hard even when the regions are pairwise disjoint. He provides several constant factor approximation algorithms for the special case of disjoint disks in the plane.  Dorrigiv et al. \cite{dorrigiv2015minimum} show that neither the minimization nor the maximization version of this problem admit an FPTAS when the regions are  disjoint disks.  Yang et al. \cite{yang2007minimum} give a PTAS for the minimization version.  In a thesis by Montanari \cite{montanari2015computing}, it is  shown that the minimization version when the input regions are vertically or horizontally aligned segments is NP-Hard and that this problem does not admit a FPTAS.  Interestingly, in another paper by Liu and Montanari \cite{liu2015minimizing} it is shown that selecting a point from each region so that {\em diameter} of a minimum spanning tree on the selected points is minimized is polynomially solvable when the regions are arbitrary sized (possibly overlapping) disks in the plane.

Montanari \cite{montanari2015computing} also studies the problem of placing a single point within each region so that the resulting shortest $s, t$ path is either maximized or minimized.  They show that the minimization version of this problem can be solved in polynomial time in the $L_1$ metric when the polygons are rectilinear (not necessarily disjoint, or convex).  They also show that the maximization version of the problem is NP-Hard to approximate to any factor $(1 - \epsilon) : \epsilon < 1/4$ even in the case where the polygons are vertically aligned segments.
 
\full{Considering regions consisting of discrete points, Daescu et al. \cite{daescu2010np}  show that, given $n$ points in the plane grouped into $m$ disjoint subsets, it is NP-Hard to select a single point from each subset so that the resulting minimum spanning tree has minimum/maximum length.  In the same paper they also show it is NP-Hard to select a single point from each subset so that the resulting convex hull has minimum perimeter, and give a $\pi$-approximation algorithm for this problem.  }

There has been a considerable amount of work done on studying TSP variants with point-existential uncertainty.  Two main models in the literature are the {\em a priori model} proposed by Bertsimas et al. \cite{bertsimas1990priori} and Jaillet~\cite{jaillet1988priori}, in which each point $x_i$ (with a known, fixed, location) is independently present with probability $p_i$, and the {\em universal model} \cite{jia2005universal}, which asks for a tour over the entire data set such that for {\em any} subset of active requests, the master tour restricted to this active subset is {\em close} to optimal.

\full{
In the {\em a priori} model, Jaillet~\cite{jaillet1988priori} derives a closed form expression for the expected tour length under a Bernoulli process, and several discrete distributions over the node presences.  He shows that there is an upper bound for the heuristic of taking the TSP on the entire point set under node-invariant probability distributions.  He shows that for general probability distributions this heuristic can be arbitrarily bad.  Further, he considers the case when the point set lies in convex position.  He shows, in this case, that the optimal a priori solution and the TSP on the entire point set are the same. Shalekamp et al. \cite{schalekamp2008algorithms} show  that when the distribution satisfies a tree metric, they can solve both, the a priori and the universal problem, in polynomial time.  As a corollary due to Fakcharoenphol et al. \cite{fakcharoenphol2003tight} any metric can be probabilistically embedded in a tree metric with $O(log(n))$ stretch, and therefore the a priori TSP can be $O(\log(n))$ approximated without knowledge of the probability distribution.  Later, Schmoys et al. \cite{shmoys2008constant} show the first constant factor approximation algorithms for the a priori TSP problem.  They show a deterministic 8-approximation algorithm and a randomized 4-approximation algorithm.

Jia et al. \cite{jia2005universal} first proposed the Universal approximation model for the TSP, Steiner tree and Set Cover problems.  They show an $O(\log^4(n)/\log(\log(n)))$-approximation algorithm in general metric spaces and an $O(\log(n))$-approximation algorithm for metric spaces with bounded doubling dimension for both the universal TSP and Steiner problems.  Hajiaghayi et al. \cite{hajiaghayi2006improved} gave the first family of examples (even for points in the plane) for which there is no constant competitive ratio for the universal TSP.  Their examples show a $\Omega(\sqrt[4]{\log(n)/\log(\log(n))})$ lower bound for the universal TSP independent of an algorithm.}

The TSP with neighborhoods (TSPN) problem was introduced by Arkin and Hassin~\cite{arkin1994approximation} and has been studied extensively from the perspective of approximation algorithms, particularly in geometric domains (see, e.g.,\cite{m-spn-04}).
Kamousi et al \cite{kamousi2013euclidean} study a stochastic TSPN model where each client lies within a region, a disk with a fixed center and stochastic radius.  
\full{They show that in the offline version, where centers are given along with the probability distributions for each disk radius one can compute an $O(\log\log(n))$-approximation, however in the online problem, when the radii of the disks are only revealed when the traveling salesman reaches the boundary of that disk, one can compute an $O(\log(n))$-approximation to the tour that minimizes the expected distance traveled.  
Further they show that if the centers as well as the mean for each radius is given, one can compute an $O(1)$-approximation.}

\subsubsection*{Preliminaries}
We are given regions $R=\{R_1,R_2,\ldots,R_n\}$, with each $R_i$ a
subset of a metric space $(X,d)$.  We seek a cyclic permutation
$\sigma=(\sigma_1,\sigma_2,\ldots,\sigma_n)$ (an {\em ordering}) of the regions $R$, in order to minimize
the length, $\max_{p_i\in R_i}
[d(p_{\sigma_1},p_{\sigma_2})+d(p_{\sigma_2},p_{\sigma_3})+\cdots+d(p_{\sigma_{n-1}},p_{\sigma_n})+d(p_{\sigma_n},p_{\sigma_1})]$,
of a cycle on adversarial choices of the points in the respective
regions.  We let $\sigma_R^*$ denote an optimal ordering for $R$, and
we let $|OPT|$ denote the length of the corresponding cycle, $OPT$, that is
based on the optimal adversarial choices of the points $p_i\in R_i$,
for the ordering $\sigma_R^*$.
\full{The following lemma gives upper/lower bounds on the length, $|OPT|$, of $OPT$.}
\short{The following lemmas are shown in the full paper~\cite{full}.}

\begin{lemma}
The length, $|OPT|$, of $OPT$ satisfies $TSPN^* \leq |OPT| \leq
TSPN^*+\sum_{R_i \in R} 2\cdot diam(R_i)$, where $TSPN^*$ is the
length of an optimal TSPN tour on the regions $R$, and $diam(R_i)$
denotes the diameter of region $R_i\in R$.
\label{lem:bounds}
\end{lemma}

\full{
\begin{proof}  
The fact that $|OPT|\geq TSPN^*$ follows immediately from the fact
that $OPT$ is a feasible solution to the TSPN. Let $p_i\in R_i$ be the
visitation points (one per region) selected by the adversary, through
which $OPT$ passes, and let $q_i\in R_i$ denote a point of region
$R_i$ visited by an optimal TSPN tour.  By adding doubled copies of
the edges $(p_i,q_i)$ to an optimal TSPN tour, we obtain a tour that
visits the points $p_i$ visited by $OPT$ (selected by the adversary);
thus, $|OPT|$ is at most $TSPN^*+\sum_{R_i\in R} 2\cdot diam(R_i)$.
\end{proof}
}

\begin{lemma}
For a set $R$ of convex regions in the Euclidean plane, and any
ordering $\sigma$ of the regions $R$, any longest cycle corresponding
to an adversarial choice of points $p_i\in R_i$ is a polygonal cycle,
with edges $(p_{\sigma_i},p_{\sigma_{i+1}})$ and with each point
$p_{\sigma_i}$ an extreme point of its corresponding region,
$R_{\sigma_i}$.
\label{lem:vertices}
\end{lemma}

\full{
\begin{proof}
Between two consecutive visitation points, the triangle
inequality implies that the cycle is a straight segment (i.e., the
cycle is polygonal, though it may self-intersect).
For the given ordering $\sigma$, the adversary selects the visitation
points $p_{\sigma_i}\in R_{\sigma_i}$ in order to maximize the length
of the cycle. 
Local optimality implies that for fixed choices of $p_{\sigma_{i-1}}$
and $p_{\sigma_{i+1}}$, the point $p_{\sigma_i}$ is chosen within
$R_{\sigma_i}$ in order to maximize
$\lambda=d(p_{\sigma_{i-1}},p_{\sigma_i})+d(p_{\sigma_i},p_{\sigma_{i+1}})$.
This implies that $R_{\sigma_i}$ must lie within the (closed) region
$A_i$ bounded by the ellipse with foci $p_{\sigma_{i-1}}$ and
$p_{\sigma_{i+1}}$ and major axis $\lambda$. The point $p_{\sigma_i}$
lies on the ellipse that bounds $A_i$; thus, the line tangent to $A_i$
at point $p_{\sigma_i}$ is a supporting line of $A_i$ and therefore
also of $R_{\sigma_i}\subseteq A_i$.  Since there is a supporting line
of $R_{\sigma_i}$ that passes through $p_{\sigma_i}$, the point
$p_{\sigma_i}$ is an extreme point of $R_{\sigma_i}$, as claimed.
\end{proof}
}

\section{3-Approximation for Arbitrary Regions in a Metric Space}

We begin by giving a $3$-approximation to the ATSP problem
when $R$ is a set of arbitrary regions in a metric space.
\full{
The main idea of this algorithm is to plan a route assuming that for
any edge of the tour we will be forced to travel the maximum distance
between consecutive regions.  Clearly, this is overestimating the
power of the adversary as he is allowed to choose only a single point
in each region after the ordering is announced, and, in general, there
is not a single point in every region that is furthest from every
other region.}

Consider the complete graph $\hat{G}$ whose nodes are the regions $R$
and whose edges join every pair of regions with an edge, $(R_i,R_j)$,
whose weight is defined to be $w(R_i,R_j)=\max_{s \in R_i, t \in R_j}\{d(s, t)\}$, the maximum
distance between a point $s\in R_i$ and a point $t\in R_j$.
For distinction, we will speak of edge ``weights'' in the graph
$\hat{G}$ and of edge ``lengths'' in the original metric space
$(X,d)$.
\short{It is not hard to see that the edge-weighted graph $\hat{G}$ defines a metric (see the full paper~\cite{full}).}

\full{
\begin{lemma}
The edge-weighted graph $\hat{G}$ is a metric.
\label{lem:metric}
\end{lemma}

\begin{proof}
Since the weight of edge $(R_i,R_j)$ is $w(R_i,R_j)=\max_{s \in R_i,
t \in R_j}\{d(s, t)\}$, we have that the edge weights of $\hat{G}$
satisfy the non-negativity and reflexivity constraints by definition.
To show that these edge weights also satisfy the triangle inequality,
consider an arbitrary triple of regions $A, B, C$.  Suppose that
$\max_{a \in A, c \in C}\{d(a, c)\}$ is maximized for $a = \hat{a} \in
A, c = \hat{c} \in C$ and that $\max_{a \in A, b \in B}\{d(a, b)\}$ is
maximized for $ b = \hat{b} \in B$. As points
$\hat{a}, \hat{b}, \hat{c}$ all come from a metric space we have that
$d(\hat{a}, \hat{c}) \leq d(\hat{a}, \hat{b}) + d(\hat{b}, \hat{c})$.
Also, by the max distance property, we have that
$d(\hat{a}, \hat{b}) \leq \max_{a \in A, b \in B}\{d(a, b)\}$, and
that $d(\hat{b}, \hat{c}) \leq \max_{b \in B, c \in C}\{d(b, c)\}$.
Therefore $\max_{a \in A, c \in C}\{d(a, c)\} \leq \max_{a \in A,
b \in B}\{d(a, b)\} + \max_{b \in B, c \in C}\{d(b, c)\}$.
\end{proof}
}

% moved to intro:
% For the remainder of the paper, let $OPT$ be the tour over the points chosen by the adversary in each region, $R_i$, in the order prescribed by an optimal permutation $\sigma_R^*$, and let $|OPT|$ denote its length.

\begin{lemma}
An optimal TSP tour in $\hat{G}$ yields a $2$-approximation to the ATSP on~$R$.
\label{lem:apx}
\end{lemma}
\begin{proof}
Let $\sigma_R^* = <R_1^*, R_2^*, ... , R_n^*>$ be an optimal (cyclic)
permutation of the regions $R$ for the adversarial TSP on $R$, and let
$p_i^*\in R_i^*$ be the adversary's choice of points corresponding to
$\sigma_R^*$.  Then, $|OPT|=d(p_1^*,p_2^*)+d(p_2^*,p_3^*)+\cdots
+d(p_n^*,p_1^*)$ is the length of the cycle
$C=<p_1^*,p_2^*,\ldots,p_n^*>$, an optimal adversarial TSP solution.

Let $w_{\sigma_R^*}=w(R_1^*,R_2^*)+w(R_2^*,R_3^*)+\cdots +
w(R_n^*,R_1^*)$ be the total weight of the cycle $\sigma_R^*$ in
$\hat{G}$. 
Let $w_{TSP}^*$ be the total weight of a minimum-weight Hamiltonian
cycle, given by (cycle) permutation $\sigma_{TSP}$, in $\hat{G}$;
then, $w_{TSP}^*\leq w_{\sigma_R^*}$.

Our goal is to show that the permutation $\sigma_{TSP}$ yields a 2-approximation for the adversarial TSP on $R$.
Since the length of the adversarial cycle corresponding to $\sigma_{TSP}$ is at most $w_{TSP}$, and since
$w_{TSP}^*\leq w_{\sigma_R^*}$, it suffices to show that $w_{\sigma_R^*}\leq 2|OPT|$.

Consider the cycle $C=<p_1^*,p_2^*,\ldots,p_n^*>$ whose length is
$|OPT|$.  If we modify $C$ by choosing points within each region
$R_i^*$ differently from $p_i^*\in R_i^*$, the length of $C$ can only
go down, since the points $p_i^*$ were chosen adversarially to make
the cycle $C$ as long as possible (for the given permutation
$\sigma_R^*$).
Consider two copies of $C$ (of total length $2|OPT|$); we will modify these
two cycles into two (possibly shorter) cycles, $C_1$ and $C_2$, by making different choices for
the points in each region $R_i^*$.

Consider first the case that $n$ is even.
Then, we define $C_1$ to be the modification of cycle $C$ in which the points are chosen in regions $R_i^*$
in order to maximize the lengths of the ``odd'' edges, $(R_1^*,R_2^*), (R_3^*,R_4^*),\ldots, (R_{n-1}^*,R_n^*)$,
and we define $C_2$ to be the modification of cycle $C$ in which the points are chosen in regions $R_i^*$
in order to maximize the lengths of the ``even'' edges, $(R_2^*,R_3^*), (R_4^*,R_5^*),\ldots, (R_n^*,R_1^*)$.
The cycle $C_1$, then, has length at least
$w(R_1^*,R_2^*)+w(R_3^*,R_4^*)+\cdots+w(R_{n-1}^*,R_n^*)$, the total
weights of the odd edges in the cycle in $\hat{G}$ corresponding to
$\sigma_R^*$.
Similarly, the cycle $C_2$ has length at least
$w(R_2^*,R_3^*)+w(R_4^*,R_5^*)+\cdots+w(R_n^*,R_1^*)$, the total
weights of the even edges in the cycle in $\hat{G}$ corresponding to
$\sigma_R^*$.
Together, then, the lengths of the two cycles $C_1$ and $C_2$ total at
least the weight, $w_{\sigma_R^*}$, of the cycle $\sigma_R^*$ in the
graph $\hat{G}$.  Since each of the weights of $C_1$ and $C_2$ are
at most $|OPT|$ (the weight of $C$), we conclude that
$w_{\sigma_R^*}\leq 2|OPT|$, as claimed.

\short{The case in which $n$ is odd is handled similarly; details appear in the full paper~\cite{full}.}
\full{Consider now the case in which $n$ is odd.
Then, we define $C_1$ to be the modification of cycle $C$ in which the
points $p_i\in R_i^*$ are chosen in regions $R_i^*$ in order to
maximize the lengths of the ``odd'' edges, $(R_1^*,R_2^*),
(R_3^*,R_4^*),\ldots, (R_{n-2}^*,R_{n-1}^*)$; then, point $p_n\in
R_n^*$ is chosen to be $p_n=a_n\in R_n^*$, the endpoint of an edge,
$(a_n,a_1)$, with $a_1\in R_1^*$, that realizes the distance
$w(R_n^*,R_1^*)$ (i.e., $d(p_n,a_1)=d(a_n,a_1)=w(R_n^*,R_1^*)$).
We define $C_2$ to be the modification of cycle $C$ in which the points $q_i\in R_i^*$ are chosen in regions $R_i^*$
in order to maximize the lengths of the ``even'' edges, $(R_2^*,R_3^*), (R_4^*,R_5^*),\ldots, (R_{n-1}^*,R_{n}^*)$;
then, subject to these choices of points $q_n\in R_n^*$ and $q_2\in R_2^*$,
we choose the point $q_1\in R_1^*$ in order to maximize $d(q_n,q_1)+d(q_1,q_2)$.
\begin{claim}
$d(q_n,q_1)+d(q_1,q_2)\geq diam(R_1^*)$, where $diam(R_1^*)$ is the diameter of the set $R_1^*$.
\end{claim}
\begin{proof}
The claim is trivially true if $R_1^*$ has only a single point (and diameter 0).
Thus, assume that $R_1^*$ has at least two points, and let $u,v\in R_1^*$ be a pair of
points that are at maximum distance (i.e., $d(u,v)=diam(R_1^*)$).
By the choice of $q_1\in R_1^*$, we know that $d(q_n,q_1)+d(q_1,q_2)\geq d(q_n,u)+d(u,q_2)$
and that $d(q_n,q_1)+d(q_1,q_2)\geq d(q_n,v)+d(v,q_2)$.  Adding these inequalities we get
$$2(d(q_n,q_1)+d(q_1,q_2))\geq d(q_n,u)+d(u,q_2)+ d(q_n,v)+d(v,q_2)$$
$$ = [d(u,q_n)+d(q_n,v)]+[d(u,q_2)+d(q_2,v)] \geq 2d(u,v).$$
This implies that $d(q_n,q_1)+d(q_1,q_2)\geq d(u,v)=diam(R_1^*)$, as claimed.
\end{proof}
The cycle $C_1$ has length at least
$w(R_1^*,R_2^*)+w(R_3^*,R_4^*)+\cdots+w(R_{n-2}^*,R_{n-1}^*) + d(p_n,p_1)$, the total
weights of the odd edges in the cycle in $\hat{G}$ corresponding to
$\sigma_R^*$, plus the distance $d(p_n,p_1)=d(a_n,p_1)$.
% (By the fact that the points in the cycle $C$ were chosen adversarially, we know that the length of $C_1$ is at most the length of $C$.)
% 
Similarly, the cycle $C_2$ has length at least
$w(R_2^*,R_3^*)+w(R_4^*,R_5^*)+\cdots+w(R_{n-1}^*,R_{n}^*)+diam(R_1^*)$, the total
weights of the even edges in the cycle in $\hat{G}$ corresponding to
$\sigma_R^*$, plus the sum of the lengths of the two edges incident on $q_1$ (which was chosen to maximize this sum).
Together, then, the lengths of the two cycles $C_1$ and $C_2$ sum to at least 
$$w(R_1^*,R_2^*)+w(R_2^*,R_3^*)+w(R_3^*,R_4^*)+\cdots+w(R_{n-2}^*,R_{n-1}^*) + w(R_{n-1}^*,R_n^*) + d(p_n,p_1) + diam(R_1^*)$$
$$\geq w(R_1^*,R_2^*)+w(R_2^*,R_3^*)+w(R_3^*,R_4^*)+\cdots+w(R_{n-2}^*,R_{n-1}^*) + w(R_{n-1}^*,R_n^*) + d(a_n,a_1)$$
$$= w(R_1^*,R_2^*)+w(R_2^*,R_3^*)+w(R_3^*,R_4^*)+\cdots+w(R_{n-2}^*,R_{n-1}^*) + w(R_{n-1}^*,R_n^*) + w(R_n^*,R_1^*).$$
Thus, the lengths of $C_1$ and $C_2$ sum to at least 
the weight, $w_{\sigma_R^*}$, of the cycle $\sigma_R^*$ in the
graph $\hat{G}$.  Since, by the adversarial choices of points in $C$, each of the lengths of $C_1$ and $C_2$ are
at most $|OPT|$ (the weight of $C$), we conclude that
$w_{\sigma_R^*}\leq 2|OPT|$, as claimed.
}
\end{proof}

%Lemma ~\ref{lem:metric} guarantees that one can find such a $3/2$-approximate TSP tour of $\hat{G}$.  Combining this with  Lemma ~\ref{lem:apx} we get the main theorem.

\full{A direct consequence of the above lemma is the following:}

\begin{theorem}
The permutation $\sigma_R$ corresponding to a Christofides $3/2$-approximate TSP tour in $\hat{G}$ yields a $3$-approximation to the adversarial TSP on~$R$.
\end{theorem}

\section{Unit Line Segments of the Same Orientation in the Plane}
In this section, we assume that $R$ consists of a set of $n$
unit-length segments of the same orientation; without loss of
generality, we assume the segments are vertical. We show that the
ordering, $TSP_c$, given by an optimal TSP tour on the segment center
points yields an adversarial tour of length at most $(7/3) |OPT| + 1$;
thus, a PTAS to approximate $TSP_c$ yields an algorithm with
adversarial tour length at most $(7/3 + \epsilon) |OPT| + 1$, for any
fixed~$\epsilon$.

% Lemma 5 in original socg submission
\begin{lemma}
\label{lem:centers}
For the ATSP on a set $R$ of unit vertical segments in the plane, 
$|OPT|\geq TSP_c^*$, where $TSP_c^*$ is the length of an optimal TSP tour on the segment center points.
\end{lemma}
\begin{proof}
Consider an ATSP optimal ordering $\sigma_R^*$ of the vertical
segments $R$. The cycle $\gamma_c$ that visits the center points of
segments $R$ in the order $\sigma_R^*$ has length at least $TSP_c^*$,
and the length, $|OPT|$, of an adversarial cycle for $\sigma_R^*$ is
at least the length of~$\gamma_c$.
\end{proof}

% Lemma 6 in original socg submission
\begin{lemma}
 $|OPT| \geq \frac{3}{4}(n-1)$ when $n$ is odd and  $|OPT| \geq \frac{3}{4}n$ when $n$ is even.
 \label{lem:sum-segs}
\end{lemma}
\begin{proof}
It suffices to show the claim for ATSP paths; an ATSP cycle is at
least as long.  The proof is by induction on $n=|R|$.  First, suppose
that $n$ is odd. The base case is trivially true.  Assume that the
claim holds for $n \leq k$, for $k$ odd.  Next, consider an instance
$S'$ with $k + 2$ segments.  We know that for the first $k$ segments
in an optimal permutation for $S'$, that $|OPT| \geq \frac{3}{4}(k -
1)$. Next, we show that regardless of the placement of the next two
unit segments in $S'$, $s_{k+1}$ and $s_{k+2}$, an adversary can make
us pay at least 3/2 units for every independent pair of consecutive
segments in $\sigma_S^*$. We can assume that (vertical) segments
$s_{k+1}$ and $s_{k+2}$ are vertically collinear. Next we assume,
without loss of generality, that $s_{k+2}$ is above $s_{k+1}$. Let $a$
be the point on $s_{k }$ that the adversary chose; refer to
Figure~\ref{fig:segs}. Let $b_1$ (resp., $c_1$) be the top endpoint of
$s_{k+1}$ (resp., $s_{k+2}$). Let $b_2$ (resp. $c_2$) be the bottom
endpoint of $s_{k+1}$ (resp., $s_{k+2}$). Now, let $|b_1a| = x$ and
$|c_2b_1| = y$. This implies that $|ab_2| = 1-x$, $|c_2b_2| = 1- y$
and $|c_1b_1| = 1-y$. The three candidate routes for the adversary to
take are $(a,b_2,c_1)$ or $(a,b_1,c_2)$ or $(a, b_1, c_1)$. These
paths have lengths $3-x-y$, $x+y$, $x+1-y$, respectively. Thus, we
solve $min-max_{x,y} \{\{3-x-y, x+y, x+1-y\} : 0\leq x \leq 1, 0 \leq
y \leq 1\}$ to find the minimum possible length of the adversarial route;
the solution is 3/2. Thus, the adversary can make us pay 3/2 for each
pair of segments; thus, $|OPT| \geq \frac{3}{4}(n-1)$ for $n$ odd.
 
In the case that $n$ is even, the adversary can make us pay at least
3/2 between every consecutive pair of segments in the optimal
ordering; thus, $|OPT| \geq \frac{3}{4}n$.
\end{proof}

\begin{figure}
\centering
\includegraphics[]{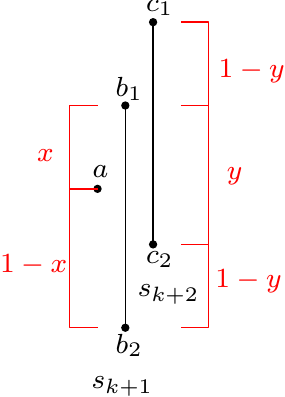}
\caption{Illustration of the induction step in the proof of Lemma~\ref{lem:sum-segs}.}
  \label{fig:segs}
\end{figure}

\begin{theorem}
For the ATSP on a set $R$ of unit-length vertical segments, the
ordering given by an optimal TSP on the segment center points yields an
adversarial tour of length at most $(7/3)|OPT|+1$.  Thus, a PTAS for
the TSP on center points yields an approximation algorithm for ATSP, 
with tour length at most $(7/3 + \epsilon) |OPT| + 1$.
\end{theorem}
\begin{proof}
Let $APX_c$ be the ordering in which the segments are visited by a
$(1 + \epsilon)$-approximate TSP tour on their center points, and let
$|APX_c|$ be the cost of the resulting adversarial tour for this
ordering. 
We know that $|APX_c| \leq |TSP_c| + n$, where $|TSP_c|$ is the length
of an optimal TSP on center points, since a tour on the center points
can be made to detour to either endpoint and back, for each segment,
at a total increase in length of $n$.
Since $|TSP_c| \leq |OPT|$ (by Lemma~\ref{lem:centers}) and $n \leq
4/3|OPT| + 1$ (by Lemma~\ref{lem:sum-segs}), we have that $|APX| \leq
(7/3 + \epsilon) |OPT| + 1$.
\end{proof}

\section{PTAS for Disjoint Unit Disks in the Plane}

In this section we give a PTAS for the adversarial TSP problem when the regions $R=\{d_1,\ldots,d_n\}$ are $n$ disjoint unit-diameter disks in the plane.  We employ the $m$-guillotine  method~\cite{mitchell1999guillotine}, which has been applied to give approximation schemes for a wide variety of geometric network optimization problems, including the Euclidean TSP and the  {\em TSP with Neighborhoods} (TSPN) when the regions are disjoint disks or fat regions in the plane~\cite{dumitrescu2001approximation,mitchell2007ptas}.  

The challenge in applying known PTAS techniques is being able to handle the adversarial nature of the tour.  For the TSPN problem, one computes (using dynamic programming) a shortest connected $m$-guillotine, Eulerian, spanning subgraph of the regions; a tour visiting each region can then be extracted from this network.  A structure lemma shows that an optimal TSPN solution can be converted to an $m$-guillotine solution whose weight is at most $(1 + \epsilon)|OPT|$.  Since $m$-guillotine networks have a recursive structure, we can apply dynamic programming in order to find the cheapest such structure over the input. Then, by extracting a tour from the optimal $m$-guillotine network, we obtain a permutation of the input disks, as well as a {\em particular} point within each region that the tour visits.

For the ATSP problem, we require new ideas and a new structure theorem to account for the fact that our algorithm must search for a permutation of the input disks that is good with respect to an adversarial path through the ordered disks.  We seek to optimize a network that has a recursive structure (to allow dynamic programming to be applied) and that yields an ordering of the disks so that the length of the adversary's tour is ``very close'' to optimal among all possible permutations.  We do this by searching for a shortest (embedded) network having an $m$-guillotine structure that has additional properties that guarantee that the adversary's path through the sequence of regions we compute is not much longer than that of the network we compute.  To accomplish this, we will require several structural results about an optimal solution to ATSP.

\begin{lemma}
\label{lem:adversary-path}
Given any ordering $\sigma_R$ of the input disks $R$, the adversarial path/cycle associated with $\sigma_R$ is a polygonal path/cycle whose vertices lie on the boundaries of the disks $d_i\in R$.
\end{lemma}

\begin{proof}
The adversary selects exactly one visitation point $p_i\in d_i$ within each disk $d_i\in R$ in order to maximize the length of the path/cycle associated with the order $\sigma_R$.  Between two consecutive (in the order $\sigma_R$) visitation points, the triangle inequality implies that the adversarial path/cycle is a straight segment (i.e., the path/cycle is polygonal).  If the adversary had chosen a visitation point $p_i$ to be interior to $d_i$, then we get a contradiction to the fact that the adversary chose visitation points to maximize the length of the associated path/cycle:  An interior point $p_i$ could be moved a nonzero amount within $d_i$ in such a way that the two edges of the path/cycle incident on $p_i$ both increase in length, e.g., by moving $p_i$ along the angle bisector of the two incident edges, in the direction opposite to the convex cone that the edges define.
\end{proof}

\subsection{Discretization and a Structural Theorem}

In order to make our problem and our algorithm discrete, for a fixed
integer $m=O(1/\epsilon)$, we place $m$ {\em sample points} evenly
spaced around the boundaries of each of the $n$ disks $d_i\in R$.  Let
$\mathcal{G}$ be the set of all $nm$ sample points.  Let
$E_{\mathcal{G}}$ denote the set of edges (line segments) between two
sample points of $\mathcal{G}$ that lie on the boundaries of different
disks of $R$.  The following lemma shows that for any adversarial
(polygonal) tour $T$ associated with $\sigma_R$ there is a polygonal
tour $T'$ visiting the sequence $\sigma_R$ whose vertices are among
the sample points $\mathcal{G}$ and whose length is at least $(1 -
O(1/m))|T|$.
%
% We begin with a lemma characterizing the structure of an adverarial path/tour.

\begin{lemma}
\label{lem:sample-points}
Given an adversarial (polygonal) path/cycle, $T$, associated with a
sequence $\sigma_R$ of input disks, there is a polygonal path/cycle
$T'$ that visits sample points $\mathcal{G}$, exactly one per disk, in
the order $\sigma_R$, such that $|T| \le (1 + O(1/m))|T'|$.
\end{lemma}

\begin{proof}
We let $T'$ be the path/cycle obtained from $T$ by rounding each of its vertices to the closest sample point of the associated (unit-diameter) disk. This rounding results in each edge decreasing in length by at most $2\cdot \frac{\pi}{m}$, since the sample points are spaced on the disk boundary at distance (along the boundary) of $\frac{2\pi(1/2)}{m}$.  Thus, $|T| \le |T'| + \frac{2\pi}{m}n$.  
We obtain a lower bound on $|T'|$, in terms of $n$, using an area argument (as done in \cite{dumitrescu2001approximation}, but included here for completeness).
Let $A(T')$ be the area swept by a disk of radius 1 whose center traverses $T'$; it is well known that the area swept by a disk of radius $\delta$ whose center moves on a curve of length $\lambda$ is at most $2\delta\lambda+\pi \delta^2$, implying that $A(T')\leq 2|T'|+\pi$. Since $T'$ meets all $n$ of the unit-diameter disks $d_i$, we know that $A(T')\geq n\cdot \pi(1/2)^2$. Thus, $n\leq (8/\pi)|T'| + 4\leq O(|T'|)$ (assuming that $|T'|\geq c$, for some constant $c$, which holds if $n\geq 2$). Since $n\leq O(|T'|)$, the inequality $|T| \le |T'| + \frac{2\pi}{m}n$ implies that $|T| \le (1+O(1/m))|T'|$, as desired. 
\end{proof}

A corollary of Lemma~\ref{lem:sample-points} is that, for purposes of obtaining a PTAS, it suffices to search for an optimal adversarial tour in the discrete graph of edges $E_{\mathcal{G}}$ on sample points.

For two consecutive disks, $d_i$ and $d_{i+1}$ in an ordering
$\sigma_R$, we refer to the convex hull of $d_i$ and $d_{i+1}$ as the
{\em fat edge} associated with $(d_i,d_{i+1})$.  The collection of
such fat edges will be called the {\em convex hull tour} associated
with $\sigma_R$.
  
\begin{theorem}
No point $p\in \mathbb{R}^2$ in the plane lies within more than a constant
number of fat edges of the convex hull tour, $OPT$, associated with an
optimal ordering $\sigma_R^*$.
\label{thm:depth}
\end{theorem}

\begin{proof}
Consider an arbitrary point $p \in \mathbb{R}^2$ and consider its
intersection with the convex hull tour of $OPT$.  Center a disk $D_p$,
of radius $K$ centered at $p$, with $K = O(1)$ a constant to be
determined later.  Since the disks $d_i$ are disjoint, there are only
a constant number ($O(K^2)$) that intersect $D_p$. We remove from $R$
these disks, as well as the (at most two) disks adjacent to them in
the tour $OPT$. Let $R'$ be the remaining set of disks after these
(constant number of) disks are removed from~$R$.

We claim that $p$ is contained in no more than a constant number of
the fat edges of $OPT$ joining two disks of $R'$.  Assume to the
contrary that more than a constant number of remaining fat edges of
the convex hull tour of $OPT$ connecting disks of $R'$ contain $p$.
Consider two such fat edges, $(d_1, d_2)$ and $(d_3, d_4)$, containing
$p$ in the region where they properly cross.  Each of these fat edges
must pass ``nearly'' diametrically across $D_p$.  That is they must
cross $D_p$ in such a way that they contain its center point $p$.  We
will show that by uncrossing these two fat edges we obtain a strictly
shorter adversarial tour, thereby contradicting the optimality
assumption. Suppose, without loss of generality, that, in order to
preserve connectivity, the uncrossing replaces $(d_1, d_2)$ and $(d_3,
d_4)$ with $(d_1, d_4)$ and $(d_3, d_2)$.  Let $v_i$ be the point of
intersection closest to $d_i$ where the adversarial edge incident on
$d_i$ crosses the boundary of $D_p$.  There are two cases.

Case 1: First suppose that $\angle v_1 p' v_4 = \angle v_3 p'
v_2 \leq \pi/2$, where $p'$ is the point where the adversarial edges
correspond to $(d_1,d_2)$ and $(d_3,d_4)$ cross; refer to Figure
~\ref{fig:swap1}.  Note that we could delete the portions of
adversarial edges $(v_1, v_2)$, and $(v_3, v_4)$ crossing $D_p$ and
replace these with the two portions of the circumference of $D_p$
connecting points $v_1, v_4$ and $v_2, v_3$ (see Figure
~\ref{fig:swap1}).  In deleting the portions of the adversarial edges
which intersect the interior of $D_p$, we saved at least $4\sqrt{K^2 -
1}$.  This value comes from the fact that the adversarial edge is
contained within the fat edge connecting these two disks, which needs
only ``nearly'' pass diametrically across $D_p$; it could be the case
that $p$ is contained within a fat edge on its boundary.  In replacing
the deleted portions of adversarial edges with two arcs comprising at
most half of the circumference of $D_p$ (with arc length at most $\pi
K$) we still get an overall savings of at least $4\sqrt{K^2 - 1} - \pi
K$.  Thus, we need to choose $K$ so that $4\sqrt{K^2 - 1} - \pi K \geq
9$ implying $K \geq 11 = O(1)$.  We will show later that this savings
of 9 units of tour length is more than enough to compensate for the
adversarial increase in the new proposed ordering.

\begin{figure}
\begin{subfigure}{0.5\textwidth}
\vspace{20mm}
\includegraphics[scale = 0.45]{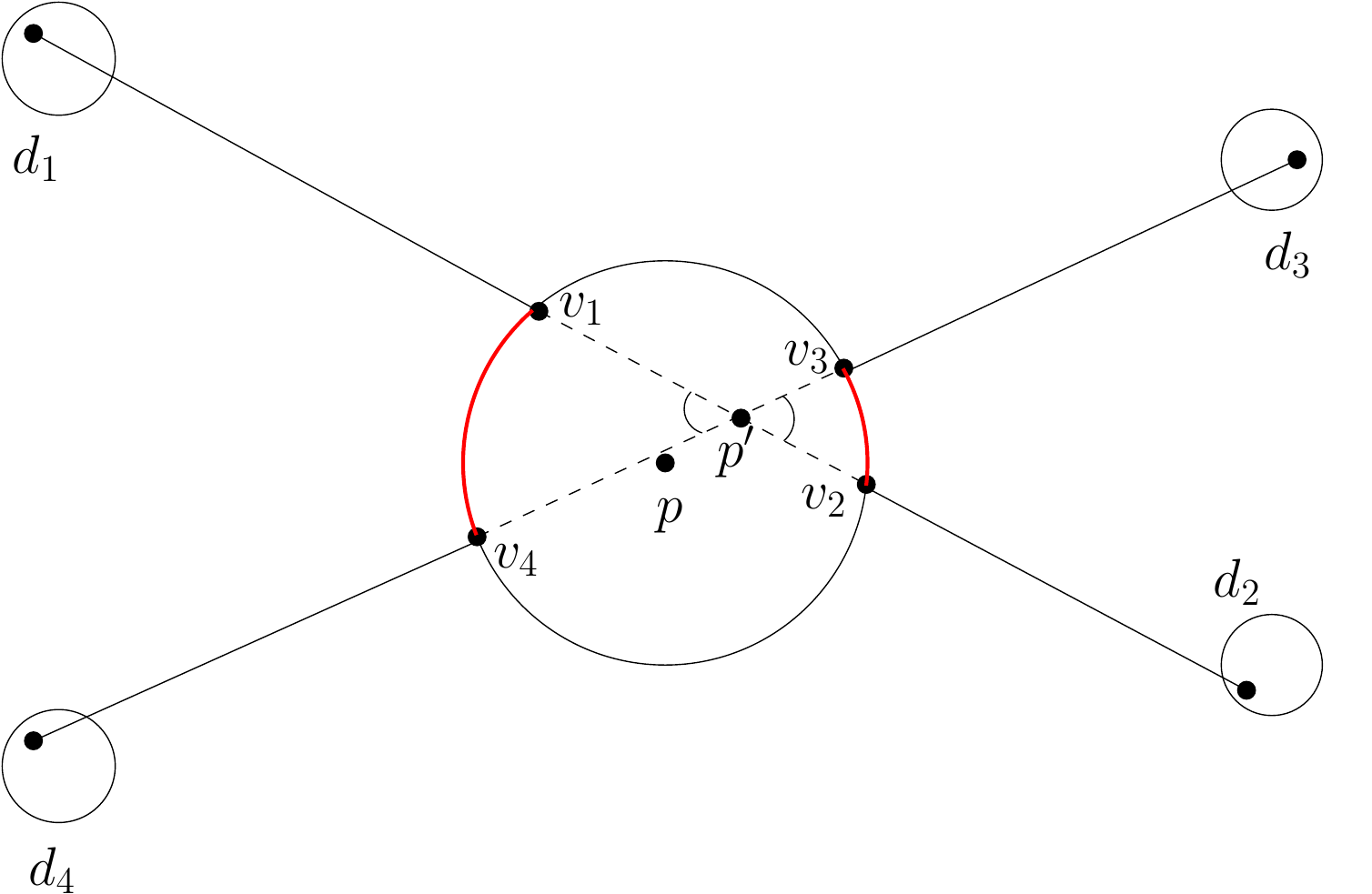}
\caption{Case 1.}
\label{fig:swap1}
\end{subfigure}
\begin{subfigure}{0.5\textwidth}
\includegraphics[scale = 0.45]{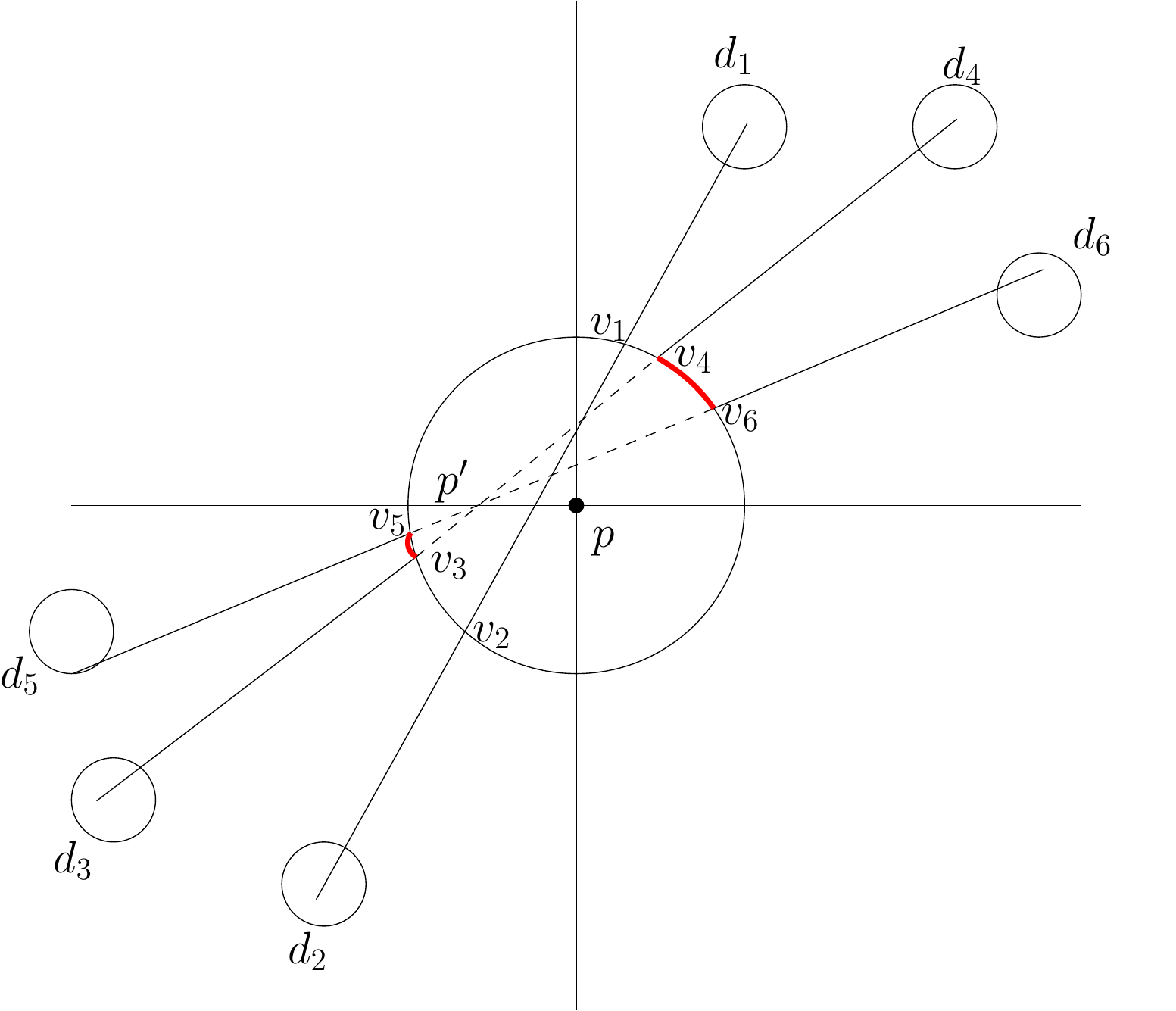}
\caption{Case 2.}
\label{fig:swap2}
\end{subfigure}
\caption{Case analysis for fat edge swapping.}

\label{fig:edge-swaps}
\end{figure}

Case 2: Next, suppose that $\angle v_1 p' v_4 = \angle v_3 p' v_2 > \pi/2$ for any pair of fat edges still containing $p$. We will begin by breaking the plane into quadrants whose origin is $p$ and now consider triples of fat edges that contain $p$.  We will only consider those triples of fat edges whose disk endpoints lie in quadrants I, and III, as we can repeat this process a finite number of times, each with a new perturbed (rotated) set of quadrants who's origin is $p$ so that eventually all remaining fat edges containing $p$ have this property.

Let $(d_1, d_2)$ be some remaining fat edge containing $p$ whose disk endpoints are in quadrants I, and III. Let $(d_3, d_4)$, $(d_5, d_6)$ be the second and third fat edge respectively that contain $p$, and have an endpoint in each of quadrants I and III, found in order by walking along the optimal tour from $d_2$ away from $d_1$.   As in case 1, let $v_i$ be the point of intersection of the adversarial edge emanating from disk $d_i$ and the boundary of $D_p$. We have that all of the $v_i$ are in quadrants I or III as well (see Figure ~\ref{fig:swap2}). Given that $v_1, v_2$ are in opposite quadrants, as well as points $v_3, v_4$, and $v_5, v_6$, a simple case analysis will show that we can delete two edges $(v_i, v_{i + 1})$, $(v_j, v_{j + 1})$ that cross the interior of $D_p$,  and replace them with two arcs of $D_p$, lying strictly within quadrants I and III, which make up at most half the circumference of $D_p$, while preserving connectivity of the tour.   This case analysis is independent of the specifics of which quadrant contains disk $d_i$, and only requires that each triple of edges we try to uncross go between opposite quadrants.

Thus, as in Case 1, we can argue that in replacing two edges crossing $D_p$ (saving at least $4\sqrt{K^2 - 1}$ in length) and replacing these with the two arcs of $D_p$ (which comprise at most half the circumference of $D_p$) we have a net savings of at least $4\sqrt{K^2 - 1} - \pi K$, which is at least 9 when $K \geq 11$. 
  
Each round of uncrossing (Case 1 or Case 2) reduces the tour length by a positive amount and reduces the depth of $p$ by at least one.  Therefore, this process will terminate in a finite number of rounds.  The number of fat edges  containing $p$ remaining after the process (Case 2) terminates will be at most (another) constant.

Finally, we argue that the constant 9 we save in tour length in each local uncrossing is enough  to compensate for whatever global increase in adversarial tour length may occur due to the new proposed ordering (since the adversary gets to re-optimize his selection of points).  

Again, consider an uncrossing of the original, hypothesized optimal tour, replacing $(d_1, d_2)$ and $(d_3, d_4)$ with $(d_1, d_4)$ and $(d_2, d_3)$.  Let $x, y, u, v$ be the (original) points adversarially chosen in disks $d_1, d_2, d_3, d_4$. After performing the uncrossing, we get a new tour, and thus the adversary gets to re-optimize by choosing a different set of points.  From the adversarial property of the initial solution, we have that the initial paths from $x$ to $u$ and from $v$ to $y$ were as long as possible over the intermediate choice of disks if we fix points $x, y, u, v$.   The new path chosen between disks $d_1$ and $d_3$ is at most that of the original path between $x$ and $u$, plus two diameters, one per disk.  That is, suppose the adversary chose new points $x', u'$ in disks $d_1, d_3$ respectively.  We can model the new path as traveling from $x'$ to $x$ in $d_1$ following the original path from $x$ to $u$ and then traveling from $u$ to $u'$ in $d_3$ costing at most two diameters.  Similarly for the path between $v$, and $y$.  Finally, in arguing about the additional length reconnecting the tour after the swap, we can upper bound, by triangle inequality, the length of the edge $(x', v')$  and $(u', y')$ as at most four diameters, one per disk $d_1, d_2, ... , d_4$ the portions of the edges $(x, y)$ $(u, v)$ strictly exterior to $D_p$ as well as at most half the circumference of $D_p$.  However, we have the savings of removing those portions of edges $(x, y)$ and $(u, v)$ that were strictly interior to $D_p$.  Recall that the diameter of $D_p$ was chosen such that removing two edges that pass ``nearly'' diametrically across $D_p$ and replacing them with two arcs comprising at most half of its circumference results in a net savings of 9 units.  Therefore in adding at most 8 diameters (or 8 units) upper bounding the adversarial increase, we still have a net savings of at least 1 unit.  Thus, we have a strictly shorter adversarial tour after performing the uncrossing, thereby contradicting the optimality assumption of the original tour. 
\end{proof}

\subsection{The $m$-Guillotine Structure Theorem}

We begin with some notation largely following \cite{dumitrescu2001approximation,mitchell1999guillotine}.  Let $G$ be an embedded planar straight line graph (PSLG) with edge set $E$ of  total length $L$, and let $R=\{d_1,\ldots,d_n\}$ be a set of disjoint unit-diameter disks $d_i$ in the plane. (In our setting, there will be exactly one vertex of $G$ within each disk $d_i\in R$.)  Let $\mathcal{B}$ be an axis-aligned bounding square of $R$.  We refer to an axis-aligned box $W \subset \mathcal{B}$ as a {\em window}, which will correspond to a particular subproblem of our dynamic program. We refer to an axis-parallel line $\ell$ that intersects window $W$ as a {\em cut} of window~$W$.  

Consider a cut $\ell$ for window $W$; assume, without loss of generality, that $\ell$ is vertical. The intersection $\ell \cap (E \cap W)$ of $\ell$ with the edge set contained in $W$ consists of a, possibly empty, set of subsegments (which include, as a degenerate case, singleton points) along $\ell$.  We let $\xi$ be the number of endpoints of subsegments along $\ell$, and let these endpoints along $\ell$ be denoted by $\beta_1, \beta_2,\ldots, \beta_\xi$ ordered by decreasing $y$ coordinate.  For a positive integer $m$ we define the {\em $m$-span} $\sigma_m(\ell)$ of $\ell$ to be $\emptyset$ if $\xi \le 2(m - 1)$, and the possibly zero length segment $\beta_m, \beta_{\xi - m + 1}$, joining the $m$th and the $m$th from the last endpoints along $\ell$ otherwise.

\newcommand{\barr}[1]{\overline{#1}}

The intersection of $\ell \cap R \cap W$ consists of a possibly empty set of $\xi_R \le |R \cap W|$ subsegments of $\ell$, one subsegment for each disk (bounding box) intersected by $\ell\cap W$.   Let these disk/boxes be $d_1, d_2,\ldots, d_{\xi_R}$ in order of decreasing $y$ coordinate.  For a positive integer $m$ we define the {\em $m$-disk-span} $\sigma_{m, R}(\ell)$ of $\ell$  to be the (possibly empty) line segment joining the bottom endpoint of $d_m \cap \ell$ to the top endpoint of $d_{\xi_R - m + 1} \cap \ell$.
In fact, as observed in \cite{mitchell2007ptas}, it suffices to consider the $m$-disk-span of the set of axis-aligned bounding squares of the input disks, since the charging scheme charges the perimeters of the regions, which are, within a constant factor, the same whether we deal with circular disks or square ($L_\infty$) disks. 
\full{(While the bounding boxes of circular disks may partly overlap, they do have constant depth, which is all that is required.)}

As in \cite{dumitrescu2001approximation} we define a line (cut) $\ell$ to be an {\em $m$-good cut} with respect to $W$ if $\sigma_m(\ell) \subseteq E$ and $\sigma_{m, R} \subseteq E$.  Finally, we say that $E$ {\em satisfies the $m$-guillotine property with respect to $W$} if either (1) $W$ does not fully contain any disk; or (2) there exists an $m$-good cut $\ell$ that splits $W$ into $W_1$, and $W_2$ and, recursively, $E$ satisfies the $m$-guillotine property with respect to $W_1$, and $W_2$.
The following is shown in \cite{dumitrescu2001approximation}, using a variant of the charging scheme of \cite{mitchell1999guillotine}:

\begin{theorem}
\label{thm:mguil}
[\cite{dumitrescu2001approximation}] 
Let $G$ be an embedded connected planar graph with edge set $E$ of total length $L$, and let $R$ be a given set of pairwise-disjoint equal-radius disks (of radius $\delta$) each of which intersects $E$.  Assume that $E$ and $R$ are contained in the square $\mathcal{B}$.  Then for any positive integer $m$ there exists a connected planar graph $G'$ that satisfies the $m$-guillotine property with respect to $\mathcal{B}$ and has edge set $E' \supseteq E$ of length $L' \leq (1+O(1/m))L + O(\delta/m)$.
\end{theorem}

In the constructive proof of Theorem~\ref{thm:mguil}, $m$-spans are added to $E$, whose lengths are charged off to a small fraction ($O(1/m)$) of the length $L$ of $E$.  Consider the edges of $E$ that cross an $m$-span, $ab$ that is added: By Theorem~\ref{thm:depth} we know that the associated fat edges (of width 1) have constant depth.  This implies that the number of edges of $E$ that cross an $m$-span, $ab$, that arises in the constructive proof of Theorem~\ref{thm:mguil} is $O(|ab|)$.

% this is a corollary to the depth theorem:
\begin{theorem}
\label{thm:depth-cor}
In the graph $G'$ that is obtained from $G$ according to Theorem~\ref{thm:mguil}, the segments of $E'$ that arise as $m$-span edges for the input edges $E$ are such that the number of edges of $E$ intersecting an $m$-span edge $ab$ is at most $O(|ab|)$.
\end{theorem}

Provided that the input $R$ is nontrivial ($n\geq 2$), the length $L^*$ of an optimal solution $OPT$ (path or cycle) to ATSP is at least 2; thus, Theorem~\ref{thm:mguil} shows that there exists an $m$-guillotine supergraph of $OPT$ of length $L'\leq (1+O(1/m))L^*$.
Further, as shown in \cite{dumitrescu2001approximation,mitchell1999guillotine,mitchell2007ptas}, one can make the $m$-guillotine conversion using cuts whose coordinates are from among a discrete set of $O(n)$ candidate $x$- and $y$-coordinates, for fixed $m$.
We will show how to use this fact, along with the structure of an optimal adversarial solution, to construct via dynamic programming an $m$-guillotine structure from which we can extract an approximation to $OPT$, with approximation factor $(1 + \epsilon)$, for any $\epsilon > 0$.  (Here, $m=O(1/\epsilon)$.)

\subsection{The Dynamic Program}

A subproblem of our dynamic program (DP) is responsible for computing a shortest total length connected network that spans the input set $R$ of disks (at their sample points) while satisfying a constant-size, $O(m)$, set of boundary conditions.  The boundary conditions specify $O(m)$ disks that the subproblem is responsible for interconnecting, as well as conditions on how the computed network within this subproblem should interact with optimal solutions computed within abutting subproblems.  As we cannot afford to keep track of all (potentially $\Omega(n)$) interconnections of the optimal ATSP solution, $OPT$, between two rectangles that bound subproblems, the $m$-guillotine structure theorem, together with our additional structural results, allow us to compactly summarize the interconnection information well enough to ensure approximation within factor $(1+\epsilon)$ of optimal.

Unlike the PTAS for TSPN, where the DP can choose any point within each region of $R$, in computing a minimum-weight connected, Eulerian, $m$-guillotine spanning subgraph over $\mathcal{G}$, in the ATSP we have no control over the point being spanned within each region:  Once we produce an ordering $\sigma_R$, the adversary gets to solve an offline longest path problem to choose the (``worst possible'') point $p_i\in d_i$ within each region $d_i\in R$ our tour must visit.  Thus, we need to create a minimum weight connected spanning Eulerian subgraph over $\mathcal{G}$ that satisfies the $m$-guillotine property {\em and} satisfies a certain {\em adversarial subpath property}, which allows us to show that in the resulting network computed by DP, we can extract a polygonal tour of $R$ that satisfies the adversarial property\full{ required for our model}.  In essence, we need the DP subproblems to be able to estimate (approximately) what the cost of an {\em adversarial} solution will be, if we extract from the optimized $m$-guillotine network a tour through~$R$.

In particular, each DP subproblem is specified by a window $W \subseteq \mathcal{B}$, along with the following additional information:
\begin{enumerate}
\item An $m$-span (possibly empty) on each of the 4 sides of $W$, each with a parity bit indicating whether the number of edges incident to the $m$-span from outside of $W$ is even or odd;    % (1)
\item $O(m)$ {\em specified edges}, which are the network edges crossing the boundary of $W$ that are not crossing one of the (up to 4) $m$-spans;  % (2)
\item An $m$-disk span (possibly empty) on each of the 4 sides of $W$, with a specified sample point given for the first and for the last disk along the $m$-disk span;   % (3)
\item $O(m)$ {\em specified input disks} (i.e., disks of $R$ not intersecting an $m$-disk span) intersecting the boundary of $W$; %(4)
\item A specified sample point of $\mathcal{G}$ on the boundary of each of the $O(m)$ specified input disks, where the network is required to visit the associated disks (these are the ``guessed'' positions of the adversarial visitation points for the specified disks);  %(5)
\item For each of the $O(m)$ specified input disks, we indicate whether the specified sample point of the disk is visited by the network being computed for the subproblem, and, if so, whether its degree in that network is 1 or 2. (The total degree of the sample point, using edges associated with subproblems on both sides of the cut, will be 2.)  % (6)
\item An interconnection pattern specifying the subsets of the $O(m)$ boundary elements (specified input disks, specified edges, $m$-spans, and $m$-disk spans) that form connected components within~$W$. %(7)
\end{enumerate}

There are only $O(n^4)$ choices for $W$, $n^{O(m)}$ choices for the specified edges/disks, and a constant ($O(g(m))$, for some function $g$) number of choices of the $O(m)$ bits and the interconnection patterns.  Thus, there are a polynomial number of subproblems for the DP.

A subproblem in the dynamic program requires one to compute a minimum-length $m$-guillotine network satisfying the following constraints:
\begin{description}
\item[(i)] The network is comprised of edges of the following types: (a). edges from the set $E_{\mathcal{G}}$ of edges linking a sample point of $\mathcal{G}$ on one disk to a sample point of $\mathcal{G}$ on another disk; (b). edges of type (a), $E_{\mathcal{G}}$, truncated at a (Steiner) attachment point on an $m$-span where the edge crosses the $m$-span or passes through an endpoint of the $m$-span; and (c) $m$-spans and $m$-disk spans that lie along cuts in the decomposition (recall that cuts lie along $O(n)$ discrete horizontal/vertical lines).  The attachment points and the endpoints of $m$-spans and $m$-disk spans constitute a set, $\mathcal{H}$, of {\em Steiner points}, distinct from the sample points $\mathcal{G}$ on the boundaries of the disks.
\item[(ii)] Each sample point of $\mathcal{G}$ within $W$ that is visited by the network has degree~2. 
\item[(iii)] The number of edges of type (b) (i.e.,  edges of $E_{\mathcal{G}}$ truncated at an $m$-span) incident on an $m$-span segment $ab$ is even or odd, according to whether the parity bit of the $m$-span is even or odd, so that the total sum of the degrees of the Steiner points $\mathcal{H}$ along an $m$-span is even. Further, the number of edges of type (b) incident on an $m$-span segment $ab$ is bounded by $c_0\cdot |ab|$, where $c_0$ is a constant arising from the structure Theorem~\ref{thm:depth}.
\item[(iv)] The network must be $m$-guillotine with respect to $W$, and, for each cut in the recursive partitioning of $W$, in the total length of the network we count each $m$-span twice; these doubled $m$-spans allow us to augment the resulting network to be Eulerian~\cite{mitchell1999guillotine}, and thereby to extract a tour (see below).  Further, we count the length of each $m$-disk span a constant ($O(1)$) times as well; this will allow the $m$-disk spans to be converted into adversarial subpaths visiting the set of disks that are spanned.
\item[(v)] The network must utilize the specified edges (which straddle the boundary of $W$).
\item[(vi)] The network must visit, at a sample point, each of the input disks \full{that are} interior to~$W$.
\item[(vii)] The network must visit each specified disk whose bit indicates it should be visited by the subproblem, at the specified sample point for that disk. Further, the network must visit the specified sample points for the first and last disk associated with each nonempty $m$-disk span.
\item[(viii)] The network must obey the interconnection pattern of the subproblem.
\item[(ix)] The network obeys the {\em adversarial subpath property}: Any maximal path, endpoints non-inclusive, within the network that goes through only sample points $\mathcal{G}$ is a longest path through the sequence of disks on which the sample points lie (one per disk).
\end{description}

\begin{lemma}
\label{lem:conversion} 
When an optimal tour $OPT$ is rounded to the grid $\mathcal{G}$ and then converted to become $m$-guillotine in the process that proves Theorem~\ref{thm:mguil}, the network that results from the augmentation of $OPT$ satisfies conditions (i)-(viii) at every level of the recursive process, for appropriate choices of the specified edges, disks, and interconnection patterns.
\end{lemma}

\begin{proof}
During the process that converts $OPT$ to be $m$-guillotine, according to the constructive proof of Theorem~\ref{thm:mguil}, most of the conditions hold automatically, by construction.  Edges of $OPT$ that cross an $m$-span, $ab$, do so at a point of $\mathcal{H}$ that has degree 4 (since the crossing edge is partitioned at the crossing point, becoming two truncated type-(b) edges, and the $m$-span is partitioned at the crossing point as well).  An $m$-span edge $ab$, by construction, extends between two points ($a$ and $b$, each a Steiner point) on edges of $OPT$ (each of which is thereby partitioned into two truncated type-(b) edges). Theorem~\ref{thm:depth-cor} implies that the fat edges associated with the edges of $OPT$ have bounded depth, implying condition (iii) holds.  Condition (viii) holds for the choice of interconnection pattern that is implied by $OPT$. 
\full{And, finally, the adversarial subpath property holds because of the adversarial path property of $OPT$ itself.}
\short{The adversarial subpath property holds because of the adversarial path property of $OPT$ itself.}
\end{proof}

We now discuss the enforcement of condition (ix), the adversarial subpath property, which is key to our being able to account for the adversary's choices during our optimization of the network length, assuring that, in the end, we can extract from the computed network a tour that is adversarial and not much longer than the overall network.

Let $(W,\Sigma)$ denote a subproblem associated with window $W$, where $\Sigma$ is a specification of the boundary constraints information (1)-(7). The dynamic programming recursion optimizes the partition of the subproblem $(W,\Sigma)$ into two subproblems, $(W_1,\Sigma_1)$ and $(W_2,\Sigma_2)$, by a horizontal or vertical cut line $\ell$ (intersecting $W$ and passing through one of the $O(n)$ discrete values of $x$, $y$-coordinates that define windows). Crucial to the correctness of the algorithm is that this recursion preserves the properties specified by the conditions (i)-(ix).

The objective function, $f(W,\Sigma)$, measures the total length of the network restricted to the window $W$; in particular, edges of $E_{\mathcal{G}}$ that are specified in the boundary constraints $\Sigma$ have their length partitioned and assigned to subproblems through which they pass.

The DP recursion optimizes over the choice of the cut line $\ell$ that partitions $W$ into $W_1$ and $W_2$, as well as the boundary conditions, $\Sigma_\ell$, along the cut, which will be part of the specifications $\Sigma_1$ and $\Sigma_2$.  The conditions $\Sigma_1$ and $\Sigma_2$ must be compatible with each other and with the choice of boundary conditions, $\Sigma_\ell$, across the cut $\ell$.  In particular, in order for $\Sigma_1$ and $\Sigma_2$ to be compatible with each other and with $\Sigma$, the specified edges of $E_{\mathcal{G}}$ across $\ell$ must match, as well as the $m$-span and $m$-disk span along the cut $\ell$. Further, the interconnection pattern of $\Sigma$ must specify subsets of boundary elements for $W$ that are yielded by taking the union of interconnection patterns for $(W_1,\Sigma_1)$ and $(W_2,\Sigma_2)$.

We let $\Sigma_\ell^{(R)}$ denote the partial specification of the boundary conditions $\Sigma_\ell$, in which we specify which pairs of disks from $R$ constitute the specified edges crossing $\ell$, but do not specify the actual sample points on the boundaries of these disks that define the endpoints of the edges from $E_{\mathcal{G}}$ being specified.  (In other words, $\Sigma_\ell^{(R)}$ specifies only the equivalence classes of the full set of conditions, $\Sigma_\ell$; the refinement of these equivalence classes will be specified in the optimization within the ``$\max$'' term of the recursion below.) The DP recursion is
$$f(W,\Sigma) = \min_{\ell, \Sigma_\ell^{(R)}} \{ \max_{\Sigma_\ell \in X(\Sigma_\ell^{(R)})} ( f(W_1(\ell), \Sigma_1(\Sigma_\ell)) + f(W_2(\ell),\Sigma_2(\Sigma_\ell)) ) \}$$
where the outer minimization is over choices of the cut $\ell$ and the cross-cut boundary conditions $\Sigma_\ell^{(R)}$, and the inner maximization is over choices of $\Sigma_\ell$ that are in the set $X(\Sigma_\ell^{(R)})$ of all boundary conditions across the cut $\ell$ that are refinements of the choice $\Sigma_\ell^{(R)}$, specifying precisely which sample points are utilized for each of the disks of $R$ that are involved in the specification $\Sigma_\ell^{(R)}$ (and not already specified by the ``parent'' choice, $\Sigma$, in cases in which edges crossing $\ell$  also extend outside of $W$ and have their sample points specified within $\Sigma$).  In the expression above, $W_1(\ell)$ and $W_2(\ell)$ are the subwindows of $W$ on either side of the cut $\ell$, and $\Sigma_1(\Sigma_\ell)$ and $\Sigma_2(\Sigma_\ell)$ are the corresponding boundary conditions on either side of $\ell$ that are inherited from $\Sigma$ and consistent with the conditions $\Sigma_\ell$. The fact that we maximize over the choices that the adversary can make, in all choices that cross the cut $\ell$, implies that we preserve the adversarial subpath property:

\begin{lemma}
\label{lem:compatible}
\full{The DP algorithm results in a network satisfying condition (ix), the adversarial subpath property.}
\short{The DP yields a network satisfying (ix), the adversarial subpath property.}
\end{lemma}

\subsection{Extracting an Approximating ATSP Tour}

The output of the DP algorithm is an $m$-guillotine network $G$ of minimum cost, where cost is total length, taking into account that $m$-spans are counted twice, and $m$-disk spans are counted $O(1)$ times (and are each of length at least 1). 
From the structure Theorem~\ref{thm:mguil}, we know that the total length of edges of $G$ is at most $|OPT|(1+O(1/m))$.

The fact that we accounted for the doubling of the $m$-spans in the optimization implies that we can afford to augment the edges along each $m$-span, in order that every Steiner point along an $m$-span has degree 4:  Initially, the points $\mathcal{H}$ along an $m$-span have degree 3, being either endpoints of the $m$-span (having a T-junction with an edge between two sample points of $\mathcal{G}$), or being a T-junction where an edge between two sample points of $\mathcal{G}$ is truncated, terminating on the $m$-span.  By the parity condition at the $m$-span, we know that there are an even number of T-junctions along the $m$-span, implying that we can add a perfect matching of segments along the $m$-span, joining consecutive pairs of T-junctions.  The total length of this matching is less than the length of the $m$-span, and is ``paid for'' by the doubling of the $m$-span lengths in the DP optimization.

The fact that we accounted for $O(1)$ copies of the $m$-disk spans in the optimization implies that we can afford to augment the edges of $G$ with an adversarial path through the sequence of disks stabbed by the $m$-disk span; such a path has length proportional to the length of the $m$-disk span, assuming the $m$-disk span is nontrivial in length\full{ (length at least 1)}. \full{(If the $m$-disk span has length less than 1, then it is intersecting only $O(1)$ disks, and these can each be specified as part of the subproblem, along with the $O(m)$ disks that are already being specified.)}

By the above discussion, the result of our algorithm is, then, a connected Eulerian network of length at most $|OPT|(1+O(1/m))$. From this network, we extract a tour $T$.  The tour $T$ is a cycle consisting of straight line segments joining points that are either sample points, $\mathcal{G}$, or Steiner points, $\mathcal{H}$.  Let $\pi_1,\pi_2,\ldots,\pi_k$ be the maximal subpaths along $T$ whose vertices are all sample points $\mathcal{G}$; i.e., each path $\pi_i$ has only vertices of $\mathcal{G}$ (no Steiner points $\mathcal{H}$), and every sample point of $\mathcal{G}$ that is a vertex of $T$ lies in exactly one path $\pi$.

Now, the number, $k$, of subpaths is at most the number of Steiner points $\mathcal{H}$ along the tour $T$, and this number is upper bounded by the number of Steiner points along $m$-spans in the entire network.  But, the total length of all $m$-spans is at most $O(1/m)\cdot |OPT|$, by the proof of Theorem~\ref{thm:mguil}.  This implies that $k\leq O(1/m)\cdot |OPT|$.  

The adversarial subpath property that was enforced in the dynamic programming algorithm implies that the subpaths $\pi_i$ are each adversarial -- their lengths are longest possible, for the given sequence of disks through which it passes (given that the path $\pi_i$ begins at sample point $p_i$ and ends at sample point $q_i$ as chosen by our dynamic program).  We obtain a new tour, $T'$, by chaining together the subpaths $\pi_i$, omitting any Steiner points that were along $T$.  The resulting tour $T'$ is not necessarily adversarial, but the following lemma shows that it is close to being so.

\begin{lemma}
\label{lem:adversarial}
Let $\sigma_R'$ be the order in which the disks $R$ are visited by the tour $T'$.  Then, the adversarial tour, $APX$, associated with $\sigma_R'$ has length at most $|T'|+O(1/m)\cdot |OPT|$.
\end{lemma}

\begin{proof}
For each subpath $\pi_i$ in our approximate tour $T'$ let $APX(\pi_i)$ be the adversarial path computed over the sequence of disks associated with $\pi_i$ in the final adversarial tour associated with $\sigma_R'$.  Similarly let $p_i^*$ (resp., $q_i^*$) be the point chosen in the first (resp., last) disk along $\pi_i$ in $APX$.  Assume that we have chosen the points $p_i$ (resp., $q_i$) in the first (resp., last) disk along $\pi_i$ in our extracted tour $T'$;  see Figure ~\ref{fig:path_extraction} for an illustration.

From the adversarial property of our computed solution we know that the length of the path $|\pi_i|$ computed from $p_i$ to $q_i$ is as long as possible over choices in intermediate disks.  Therefore we can over estimate the length of $APX$ by walking around $T'$ adding at most two unit diameter detours to the first and last disk in each sub-path $\pi_i$.  That is we can begin at $p_i$ detour to $p_i^*$ and back, follow $\pi_i$ until we reach $q_i$ an then detour from $q_i$ to $q_i^*$ and back and follow $T'$ to $p_{i + 1}$ and so on.  By triangle inequality and the fact that $\pi_i$ is a longest path for fixed choices of $p_i$, and $q_i$ this over estimates the length of $APX$.

We have $|APX| \leq |T'| + 4k$, and therefore $|APX| \leq |T'| + O(1/m)|OPT|$, because, as previously stated $k \leq O(1/m)|OPT|$.
\end{proof}

\begin{figure}
\centering
\includegraphics[scale = 0.5]{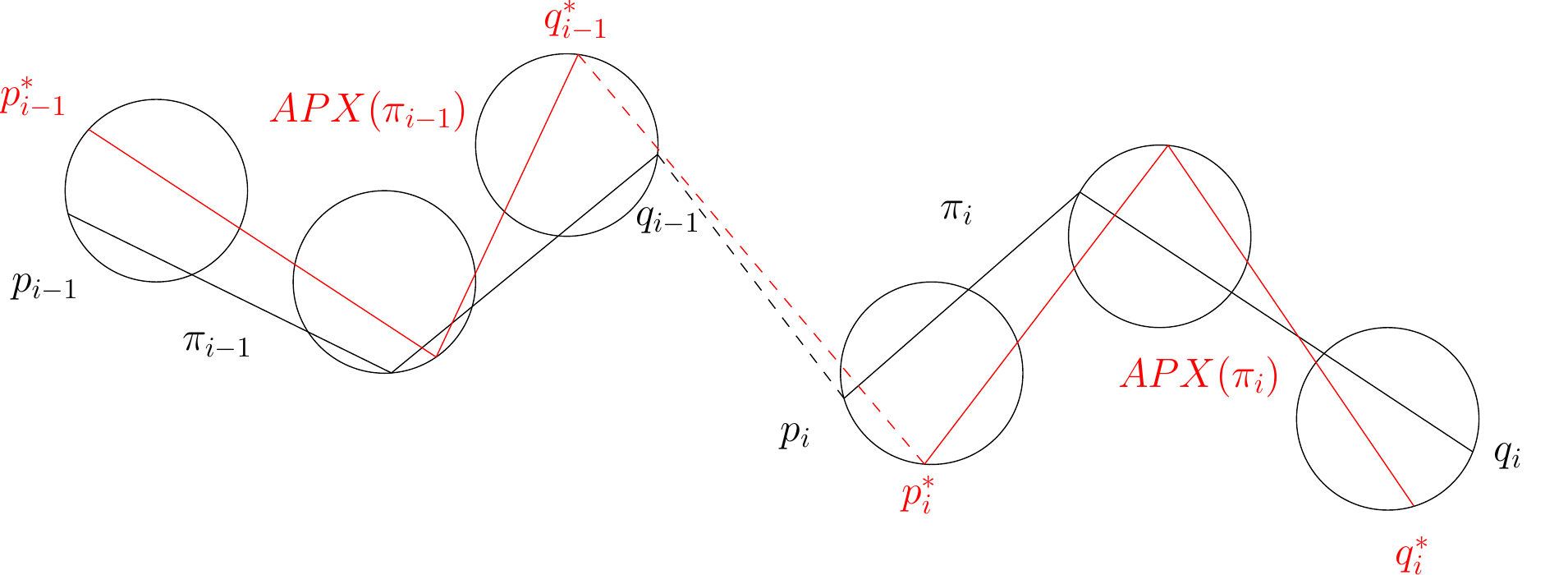}
\caption{Bounding the adversarial increase over extracted approximate tour $T'$.}
\label{fig:path_extraction}
\end{figure}

Since we know that the computed $T$, and thus $T'$, has length at most $|OPT|(1+O(1/m))$, Lemma~\ref{lem:adversarial} implies that the overall solution extracted from our computed tour yields a PTAS\full{ for ATSP}.

\full{{\em Remarks:}  (1) We have focussed here on the case of disjoint equal-radius disks in the plane.  Our methods apply more generally to the case of fat regions that are nearly the same size, with a bounded depth of overlap. (2) The case of disjoint fat regions of arbitrary sizes has a PTAS for the TSPN problem~\cite{mitchell2007ptas}; one may expect that our methods can be extended to that case as well, but our structure theorem on the bounded depth of fat edges would have to be reformulated.  (3) More efficient (possibly nearly linear time) algorithms should be possible based on applying the methods to a $t$-spanner of $\mathcal{G}$; we leave this to future work.
}

%%
%% Bibliography
%%

%% Either use bibtex (recommended), 

% \bibliography{refs}
% \bibliography{lipics-v2016-sample-article}

\begin{thebibliography}{10}

\bibitem{arkin1994approximation}
Esther~M. Arkin and Refael Hassin.
\newblock Approximation algorithms for the geometric covering salesman problem.
\newblock {\em Discrete Applied Mathematics}, 55(3):197--218, 1994.

\bibitem{bertsimas1990priori}
Dimitris~J Bertsimas, Patrick Jaillet, and Amedeo~R Odoni.
\newblock A priori optimization.
\newblock {\em Operations Research}, 38(6):1019--1033, 1990.

\short{
\bibitem{full}
Gui Citovsky, Tyler Mayer, and Joseph S. B. Mitchell.
\newblock TSP With Locational Uncertainty: The Adversarial Model.
\newblock arXiv, March, 2017.
}

\full{
\bibitem{daescu2010np}
Ovidiu Daescu, Wenqi Ju, and Jun Luo.
\newblock {NP}-completeness of spreading colored points.
\newblock In {\em Internat. Conference on Combinatorial Optimization and
  Applications}, pages 41--50. Springer, 2010.
}

\bibitem{dorrigiv2015minimum}
Reza Dorrigiv, Robert Fraser, Meng He, Shahin Kamali, Akitoshi Kawamura,
  Alejandro L{\'o}pez-Ortiz, and Diego Seco.
\newblock On minimum-and maximum-weight minimum spanning trees with
  neighborhoods.
\newblock {\em Theory of Computing Systems}, 56(1):220--250, 2015.

\bibitem{dumitrescu2001approximation}
Adrian Dumitrescu and Joseph S.~B. Mitchell.
\newblock Approximation algorithms for {TSP} with neighborhoods in the plane.
\newblock {\em Journal of Algorithms}, 48:135--159, 2003.
% \newblock Special issue devoted to 12th ACM-SIAM Symposium on Discrete Algorithms, Washington, DC, January, 2001.

\full{
\bibitem{fakcharoenphol2003tight}
Jittat Fakcharoenphol, Satish Rao, and Kunal Talwar.
\newblock A tight bound on approximating arbitrary metrics by tree metrics.
\newblock In {\em Proc. 35th ACM Symposium on Theory of Computing}, pages
  448--455. ACM, 2003.
}

\bibitem{fraser2012algorithms}
Robert Fraser.
\newblock {\em Algorithms for geometric covering and piercing problems}.
\newblock PhD thesis, University of Waterloo, 2012.

\full{
\bibitem{hajiaghayi2006improved}
Mohammad~T. Hajiaghayi, Robert Kleinberg, and Tom Leighton.
\newblock Improved lower and upper bounds for universal {TSP} in planar
  metrics.
\newblock In {\em Proc. 17th ACM-SIAM Symposium on Discrete Algorithms}, pages
  649--658. SIAM, 2006.
}

\bibitem{jaillet1988priori}
Patrick Jaillet.
\newblock A priori solution of a traveling salesman problem in which a random
  subset of the customers are visited.
\newblock {\em Operations Research}, 36(6):929--936, 1988.

\bibitem{jia2005universal}
Lujun Jia, Guolong Lin, Guevara Noubir, Rajmohan Rajaraman, and Ravi Sundaram.
\newblock Universal approximations for {TSP}, {Steiner} tree, and set cover.
\newblock In {\em Proc. 37th ACM Symposium on Theory of Computing}, pages
  386--395. ACM, 2005.

\bibitem{kamousi2013euclidean}
Pegah Kamousi and Subhash Suri.
\newblock Euclidean traveling salesman tours through stochastic neighborhoods.
\newblock In {\em Internat. Symposium on Algorithms and Computation}, pages
  644--654. Springer, 2013.

\bibitem{liu2015minimizing}
Chih-Hung Liu and Sandro Montanari.
\newblock Minimizing the diameter of a spanning tree for imprecise points.
\newblock In {\em Internat. Symposium on Algorithms and Computation}, pages
  381--392. Springer, 2015.

\bibitem{loffler2010largest}
Maarten L{\"o}ffler and Marc van Kreveld.
\newblock Largest and smallest convex hulls for imprecise points.
\newblock {\em Algorithmica}, 56(2):235--269, 2010.

\bibitem{mitchell1999guillotine}
Joseph S.~B. Mitchell.
\newblock Guillotine subdivisions approximate polygonal subdivisions: A simple
  polynomial-time approximation scheme for geometric {TSP}, {$k$}-{MST}, and
  related problems.
\newblock {\em SIAM Journal on Computing}, 28(4):1298--1309, 1999.

\bibitem{m-spn-04}
Joseph S.~B. Mitchell.
\newblock Shortest paths and networks.
\newblock In Csaba~D. T\'oth, Joseph O'Rourke, Jacob~E. Goodman, editors, {\em Handbook of Discrete
  and Computational Geometry (3rd Edition)}, chapt~31, CRC Press, 2017.
% 2nd edition: pages 607--641.   Chapman \& Hall/CRC, 2004.  % Boca Raton, FL, 2004. 
%  Updated as chapt~31, 3rd Edition, C.~T\'oth editor, 2016 (to appear).

\bibitem{mitchell2007ptas}
Joseph S.~B. Mitchell.
\newblock A {PTAS} for {TSP} with neighborhoods among fat regions in the plane.
\newblock In {\em Proc. 18th ACM-SIAM Symposium on Discrete Algorithms}, pages
  11--18. SIAM, 2007.
% \newblock URL:  \url{http://www.ams.sunysb.edu/~jsbm/papers/tspn-soda07-rev.pdf}.
%% \full  version should link to arxiv

\bibitem{montanari2015computing}
Sandro Montanari.
\newblock {\em Computing routes and trees under uncertainty}.
\newblock PhD Dissertation, ETH-Z{\"u}rich, No. 23042, 2015.

\full{
\bibitem{schalekamp2008algorithms}
Frans~Schalekamp and David~B. Shmoys.
\newblock Algorithms for the universal and a priori {TSP}.
\newblock {\em Operations Research Letters}, 36(1):1--3, 2008.
}

\full{
\bibitem{shmoys2008constant}
David~B. Shmoys and Kunal~Talwar.
\newblock A constant approximation algorithm for the a priori traveling
  salesman problem.
\newblock In {\em Internat. Conference on Integer Programming and
  Combinatorial Optimization}, pages 331--343. Springer, 2008.
}

\bibitem{yang2007minimum}
Yang Yang, Mingen Lin, Jinhui Xu, and Yulai Xie.
\newblock Minimum spanning tree with neighborhoods.
\newblock In {\em Internat. Conference on Algorithmic Applications in
  Management}, pages 306--316. Springer, 2007.

\end{thebibliography}

%% .. or use the thebibliography environment explicitely

\subsection*{Acknowledgement}
This research was partially supported by the National Science Foundation (CCF-1526406).

\end{document}